\documentclass[letterpaper, 10pt, conference]{ieeeconf}
\usepackage{graphicx} 

\IEEEoverridecommandlockouts

\usepackage{amsmath}
\usepackage{amssymb}
\usepackage{amsthm}
\usepackage{algorithm}
\usepackage{algpseudocode}
\usepackage{etoolbox} 
\usepackage[hidelinks]{hyperref} 
\usepackage[dvipsnames]{xcolor}
\usepackage{optidef}    

\newcommand{\arxiv}[2]{\ifbool{\usingarxiv}{#1}{#2}} 

\newcommand{\norm}[1]{\rVert #1\lVert}                  

{                                                       
    \theoremstyle{plain}
    \newtheorem{assumption}{Assumption}
    \newtheorem{proposition}{Proposition}
    \newtheorem{theorem}{Theorem}
    \newtheorem{lemma}{Lemma}

    \theoremstyle{remark}
    \newtheorem*{remark}{Remark}

    \theoremstyle{definition}
    \newtheorem{definition}{Definition}

}

\newcommand{\bbR}{\mathbb{R}}
\newcommand{\bbN}{\mathbb{N}}

\newcommand{\calA}{\mathcal{A}}

\newcommand{\calX}{\mathcal{X}}
\newcommand{\calU}{\mathcal{U}}
\newcommand{\calR}{\mathcal{R}}
\newcommand{\calE}{\mathcal{E}}

\newcommand{\calK}{\mathcal{K}}
\newcommand{\calT}{\mathcal{T}}

\newcommand{\calW}{\mathcal{W}}
\newcommand{\calV}{\mathcal{V}}
\newcommand{\calD}{\mathcal{D}}
\newcommand{\calF}{\mathcal{F}}
\newcommand{\calG}{\mathcal{G}}
\newcommand{\calH}{\mathcal{H}}
\newcommand{\calJ}{\mathcal{J}}
\newcommand{\calI}{\mathcal{I}}
\newcommand{\calN}{\mathcal{N}}
\newcommand{\calM}{\mathcal{M}}

\newcommand{\calL}{\mathcal{L}}
\newcommand{\calS}{\mathcal{S}}

\newcommand{\bx}{\ensuremath{\mathbf{x}}}
\newcommand{\bu}{\ensuremath{\mathbf{u}}}

\newcommand{\diam}[1]{\textup{diam}(#1)}

\newcommand{\levset}[2]{\textup{lev}_{#1}#2}
\newcommand{\set}[2]{\{#1 \: | \: #2\}}

\title{\LARGE \bf  Layered Nonlinear Model Predictive Control for\\ Robust Stabilization of Hybrid Systems}
\author{
    Zachary Olkin$^{1}$ and Aaron D. Ames$^{1}$
    \thanks{This research is supported by the Technology Innovation Institute (TII) and this material is based upon work supported by the National Science Foundation Graduate Research Fellowship.}
    \thanks{$^{1}$Authors are with the Department of Control and Dynamical Systems, California Institute of Technology, Pasadena CA 91125, U.S.A. \texttt{\{zolkin, ames\}@caltech.edu}.}
}

\date{September 2024}

\begin{document}

\maketitle

\begin{abstract}
Computing the receding horizon optimal control of nonlinear hybrid systems is typically prohibitively slow, limiting real-time implementation. To address this challenge, we propose a layered Model Predictive Control (MPC) architecture for robust stabilization of hybrid systems. A high level ``hybrid" MPC is solved at a slow rate to produce a stabilizing hybrid trajectory, potentially sub-optimally, including a domain and guard sequence.  This domain and guard sequence is passed to a low level ``fixed mode" MPC which is a traditional, time-varying, state-constrained MPC that can be solved rapidly, e.g., using nonlinear programming (NLP) tools. 
A robust version of the fixed mode MPC is constructed by using tracking error tubes that are not guaranteed to have finite size for all time.
Using these tubes, we demonstrate that the speed at which the fixed mode MPC is re-calculated is directly tied to the robustness of the system, thereby justifying the layered approach.
Finally, simulation examples of a five link bipedal robot and a controlled nonlinear bouncing ball are used to illustrate the formal results.
\end{abstract}

\section{Introduction}
Hybrid system models of robots are ubiquitous because of their ability to represent changes in contact between the robot and the environment.  The classic example is walking robots, where the feet interacting with the world result in discrete jumps in the robot's states and transitions to different underlying continuous dynamics.  This paper aims to design robust stabilizing controllers that allow robots to interact with their environment, e.g., by locomoting.
In particular, we propose a layered architecture that decomposes the controller into three distinct layers (Fig. \ref{fig:hero}): a hybrid MPC, a fixed mode MPC, and a low level controller. The hybrid MPC chooses which guards and domains to traverse. The fixed mode MPC uses these domains and guards to compute a robustly stabilizing trajectory at a fast rate. Finally, the low level controller tracks the MPC trajectory. 

\begin{figure}
    \centering
    \includegraphics[width=0.85\linewidth]{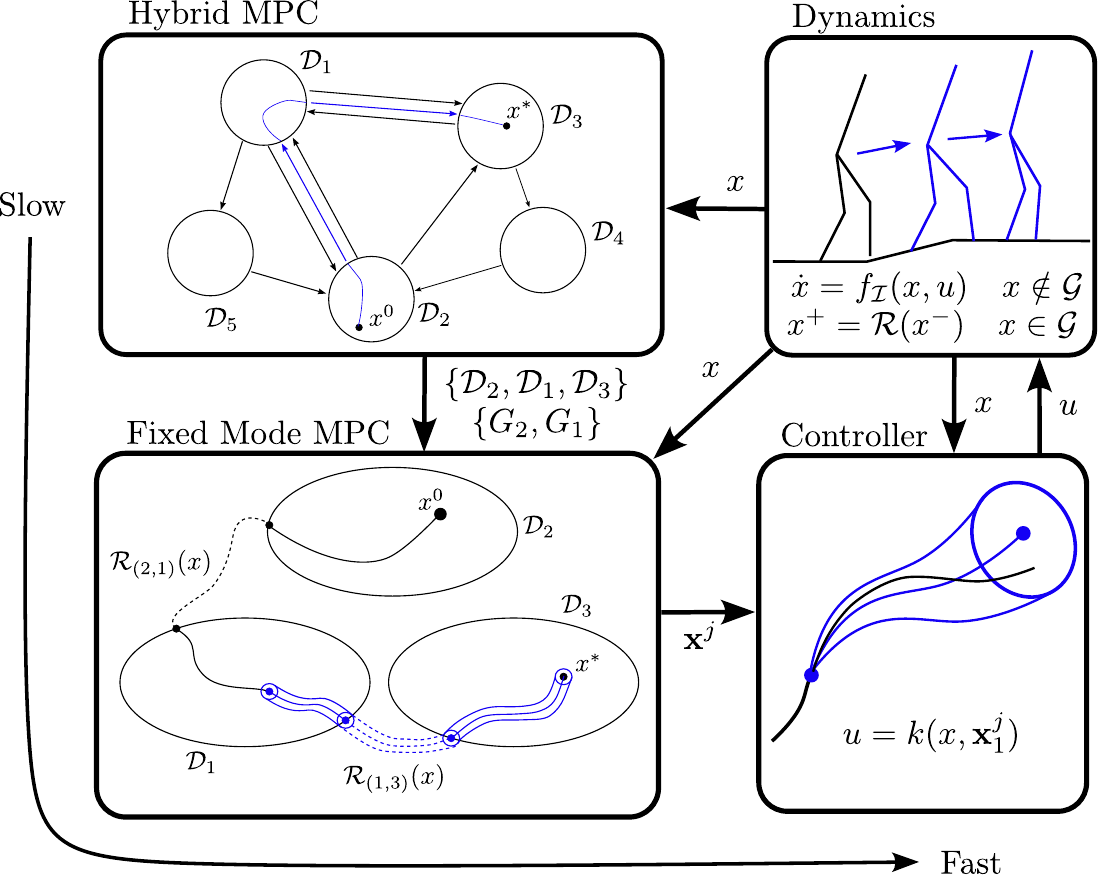}
    \caption{The hybrid MPC layer determines a feasible path at a slow rate, and passes these modes to the fixed mode MPC. This MPC computes a robustly stabilizing solution at a higher rate. The trajectory and feed forward input is passed to the low level controller where it is tracked. We use properties of the low level controller to determine the size of the tube. The resulting control action is applied to the dynamics.}
    \vspace{-8mm}
    \label{fig:hero}
\end{figure}
Computing MPC solutions for hybrid systems has been studied extensively in the field of legged robotics. Planning with the hybrid dynamics can be formulated as a mixed integer problem \cite{aceituno-cabezas_simultaneous_2017}, but this is generally too slow for real time robustness. There has also been work examining Contact Implicit MPC (CIMPC) in which the hybrid problem is solved with gradient based methods in a single optimization \cite{kong_ilqr_2021, kurtz_inverse_2023}. In general, they are slower than solving the fixed mode problem, and sometimes suffer from non-physical approximations of contact (such as force at a distance) \cite{wensing_optimization-based_2022}. By using a layered approach we can circumvent poor dynamics approximations and recover the computation speed of the fixed mode MPC. Other work has considered the problem of adjusting the switching times given a fixed mode sequence \cite{xu_optimal_2004, farshidian_sequential_2017}. A multi-layer approach was used in \cite{olkin_bilevel_2024}, and although the dynamics domains could be chosen online, the contact surfaces (guards) were fixed. A geometric approach can be used, but this smooths the reset map \cite{westenbroek_smooth_2021}.

To construct a robust MPC, we leverage the idea of a tracking error tube.
In classical tube MPC, a finite-sized tube about the trajectory is certified, and the MPC plans with knowledge of this tube to provide robustness \cite{mayne_robust_2005, singh_robust_2017}. In this scheme, the tube provides all the robustness for the system and the MPC only provides optimality and nominal stability. We take a different approach since for many robotic systems this tube cannot be certified and optimality is a secondary objective. We consider inherent robustness of MPC \cite{rawlings_model_2017} with disturbances quantified through tubes.

This paper presents a new layered MPC framework for nonlinear hybrid systems, leading to robust stabilization. The main contributions of this paper are: (a) development of a layered MPC architecture designed for hybrid systems with modifications for robustness. Stability and robustness of these constructions is analyzed. (b) Tubes with sizes that tend to infinity are quantified, including the transient effects of low level controllers. MPC computation speed is linked to the robustness of the system. (c) Simulations demonstrate the efficacy of this approach on a biped traversing stairs and a controlled bouncing ball, both subject to external forces.

\section{Preliminaries}
\subsection{Dynamical Systems}
In this paper we will consider a hybrid dynamical system:
\begin{definition}[Hybrid System]
\label{def:hybrid_sys}
$\calH = (\calJ, \Gamma, \calD, \calU, \calF, \calG, \calR)$ is a hybrid dynamical system with the following definitions:
\begin{enumerate}
    \item $\calJ \subset \bbN$ is a finite set of discrete modes. This makes up the nodes of a graph.
    \item $\Gamma \subset \calJ \times \calJ$ is a set of transitions from one mode to another that form a directed graph structure over $\calJ$.
    \item $\calD := \{\calD_{\calI} \subset \bbR^n \}_{\calI \in \calJ}$ is the collection of domains corresponding to the modes given by $\calJ$. Each domain is a connected subset of $\bbR^n$.
    \item $\calU := \{\calU_{\calI} \subset \bbR^m \}_{\calI \in \calJ}$ is the collection of compact sets of allowable control inputs.
    \item $\calF := \{f_{\calI} : \calD_\calI \times \calU_\calI\ \rightarrow \bbR^n\}_{\calI \in \calJ}$ is the collection of locally Lipschitz continuous maps defining continuous time controlled dynamical systems.
    \item $\calG := \{\calG_{(\calI, \calN)} \subset \calD_\calI\}_{(\calI, \calN) \in \Gamma}$ is the collection of guards which are each proper subsets of $\calD_\calI$. Specifically, each guard is a  $c_{\calG_{\calI, \calN}}$ level set of a continuous function $h_{\calG_{\calI, \calN}}(x)$: $\calG_{\calI, \calN} = \set{x}{h_{\calG_{\calI, \calN}}(x) = c_{\calG_{\calI, \calN}}}$. Each guard is associated with a transition in $\Gamma$.
    \item $\calR := \{\calR_{G} : \calG \rightarrow \calD\}_{G \in \calG}$ is the collection of continuous reset maps associated with each guard. There is exactly one reset map for each guard.
\end{enumerate}
\end{definition}
\begin{assumption}
\label{assump:reset_not_on_guard}
    The image of the reset maps, $\text{Im}(\calR_G)$, satisfies $\text{Im}(\calR_G) \notin \calG$.
\end{assumption}
Let Assumption \ref{assump:reset_not_on_guard} be true throughout the paper. This tells us that the reset map does not map onto any guard.

The continuous time dynamics for each domain $\calI \in \calJ$ are of the form
\begin{equation}
    \label{eq:continuous_sys}
    \dot{x} = f_\calI(x, u) + w(t)
\end{equation}
where $x \in \calD_\calI$ is the state, $u \in \calU_\calI$ in the input, and $w \in \calW \subset \bbR^n$ is a disturbance that lies within a set $\calW$. The disturbance is assumed to satisfy Assumption \ref{assump:bounded_disturbance}.
\begin{assumption}[Disturbance Set]
\label{assump:bounded_disturbance}
    The disturbance, $w(t)$ lies in some set $\calW$ where $0 \in \calW$. Further $\calW$ is bounded with $\sup_{w \in \calW}\norm{w} = \eta < \infty$.
\end{assumption}

We define the set of finite parameterizations (e.g. splines knots) of dimension $\hat{m}$ for controllers in $\calU_\calI$ as $\hat{\calU}_\calI \subset \bbR^{\hat{m}}$ with $\hat{\calU} := \{\hat{\calU}_\calI\}_{\calI \in \calJ}$. In the following we use $u \in \hat{\calU}_\calI$ to implicitly define the continuous time controller.

Since we will use MPC, we need to consider the discretization of the system given in Eq. \eqref{eq:continuous_sys}. The discrete time systems for each domain are given by
\begin{equation}
    \label{eq:discrete_sys}
    \bar{x}^{k+1} = \bar{f}_\calI(\bar{x}^k, \bar{u}^k)
\end{equation}
where $\bar{x} \in \calD_\calI$ is the state, and $\bar{u} \in \hat{\calU}_\calI$ is the input, and $k$ is the time step. This system is associated with a $\delta t_d$ that defines the integration length such that
\begin{equation}
    \label{eq:discrete_int}
    \bar{f}_\calI(\bar{x}^k, \bar{u}^k) = \int_{0}^{\delta t_d}f_\calI(x(t), \bar{u}^k(x(t), t)) dt + \bar{x}^k.
\end{equation}

\subsection{Lyapunov Methods}
To establish stability and robustness conditions, we will leverage Lyapunov methods. We briefly review both discrete and continuous Lyapunov functions. Throughout we will use $\calK$, $\calK_\infty$, and $\calK \calL$ to denote class $\calK$, $\calK_\infty$, and $\calK \calL$ functions \cite{rawlings_model_2017}.

\begin{definition} [Discrete Time Lyapunov Function]
    Consider the discrete time system $x^{k+1} = g(x^k)$. A Lyapunov function, $\calV_g(x)$, satisfies:
    \begin{align}
        \lambda_1(\norm{x^k}) \leq \calV_g(x^k) &\leq \lambda_2(\norm{x^k}) \\
        \calV_g(x^{k+1}) - \calV_g(x^{k}) &\leq -\lambda_3(\norm{x^k})
    \end{align}
    Where $\lambda_i \in \calK, \; i \in \{1,2,3\}$.
\end{definition}

If there exists a $\calV_g(x)$ then $g(x)$ is asymptotically stable to the origin and stable in the sense of Lyapunov \cite{rawlings_model_2017}.

We will use continuous time Lyapunov theory to provide tracking errors in the face of disturbances. Therefore, we will introduce the notion of Exponential Input To State Stability (E-ISS) and E-ISS Lyapunov functions. 
\begin{definition} [E-ISS]
    Consider a continuous time system, $\dot{x} = s(x) + w(t)$. Let Assumption \ref{assump:bounded_disturbance} hold. The system is E-ISS if there exists $\gamma \in \calK$, and $a \in \bbR_{>0}$ such that for every $t \geq t_0 \geq 0$ we have
    \begin{equation*}
        \norm{x(t_0)} \leq a \Rightarrow \norm{x(t)} \leq Me^{\lambda(t - t_0)}\norm{x(t_0)} + \gamma(\eta)
    \end{equation*}
\end{definition}
If a system is E-ISS, it converges to a ball about the origin and renders that ball positive invariant.

\begin{definition} [Continuous Time E-ISS Lyapunov Function]
\label{def:iss_lyap}
    Consider a continuous time system, $\dot{x} = s(x) + w(t)$. Let Assumption \ref{assump:bounded_disturbance} hold. An E-ISS Lyapunov function, $\calV_s(x)$, satisfies:
    \begin{align}
        k_1 \norm{x} \leq \calV_s(x) &\leq k_2 \norm{x} \\
        \dot{\calV_s}(x) &\leq -k_3\norm{x} + \sigma(\eta)
    \end{align}
where $\sigma \in \calK$, $k_i > 0, \; i \in \{1,2,3\}$. If there exists a E-ISS Lyapunov function, the corresponding system is E-ISS \cite{kolathaya_input_2018}.
\end{definition}
\subsection{MPC Stability and Robustness}
Throughout this paper we will only consider stability and robustness of MPC as applied to discrete time systems. Both the hybrid MPC and the fixed mode MPC cost functions are assumed to have the following form
\begin{equation}
\label{eq:cost_function_form}
    J(\bx, \bu) = \sum_{k = 0}^{N-1} L(\bx_k, \bu_k) + \Phi(\bx_N)
\end{equation}
where $L(\cdot, \cdot)$ is the stagewise cost function, and $\Phi(\cdot)$ is the terminal cost, and $\bx \in \bbR^{n\times N}$, $\bu \in \bbR^{\hat{m}\times N}$ denote state and input trajectories. At each time step in the trajectory we denote the input and states as $\bu_k \in \bbR^{\hat{m}}$, $\bx_k \in \bbR^n$ .

Without loss of generality we define the target equilibrium as $x^* = 0$ and define $\calD_F$ as the domain with $x^*$. We make the following common assumptions \cite{rawlings_model_2017}.

\begin{assumption} [Continuity of Cost]
    \label{assump:cont_cost}
    The functions $L : \bbR^n \times \bbR^{\hat{m}} \rightarrow \bbR$, and $\Phi : \bbR^n \rightarrow \bbR$ are continuous on their entire domain, $L(0,0) = 0$, and $\Phi(0) = 0$.
\end{assumption}
Define $\levset{c}{V} := \set{x}{V(x) \leq c}$ for any function $V$.
\begin{assumption} [Properties of Constraint Sets]
    \label{assump:constraint_sets}
    For all $\calI \in \calJ$, $\calU_\calI$ is compact and contains the origin. The terminal constraint, $\calE \subset \calD_F$, is given by $\calE := \levset{c_f}{\Phi} = \set{x}{\Phi(x) \leq c_f}$ for some $c_f > 0$.
\end{assumption}
\begin{assumption} [Terminal Stability]
    \label{assump:terminal_stability}
    There exists a terminal control law $\kappa_f : \calE \rightarrow \hat{\calU}_{\calD_F}$ such that for all $x \in \calE$:
    \begin{align*}
        \bar{f}(x, \kappa_f(x)) &\in \calE \\
        \Phi(\bar{f}(x, \kappa_f(x)) &\leq \Phi(x) - L(x, \kappa_f(x))
    \end{align*}
\end{assumption}
\begin{assumption} [Stage Cost Bound]
    \label{assump:stage_cost_bound}
    There exists a function $\alpha_L \in \calK_\infty$ such that $L(x, u) \geq \alpha_L(\norm{(x, u)})$.
\end{assumption}

\section{Control Architecture and Properties}
\textbf{Hybrid MPC}. This layer is responsible for determining the sequence of modes to be utilized by the control architecture. Ultimately, this amounts to determining the path used to traverse the graph structure generated by $\calJ$ and $\Gamma$. In general, this cannot be solved with normal graph traversal algorithms because we need a dynamically feasible trajectory that respects dynamics and input constraints.

The hybrid MPC problem at time step $j$ is given as
\begin{mini!}[1]
{\bx, \; \bu, \; \calI}{J^H(\bx, \bu)}
{\label{eq:hl_mpc}}{}
\addConstraint{\bx_0}{= \bar{x}^j \label{eq:ic_constraints_hl_mpc}}
\addConstraint{\bx_{k+1}}{= \bar{f}_{\calI(k)}(\bx_k, \bu_k) \quad \label{eq:dyn_constraints_hl_mpc}}{\bx_k \notin \calG}
\addConstraint{\bx_{k+1}}{= \calR_{\calG_{(\calI(k), \calI(k+1))}}(\bx_k)}{x_k \in \calG_{(\calI(k), \calI(k+1))}}
\addConstraint{\calI(k)}{\in \calJ}{\forall k \in [0, \: N]}
\addConstraint{\calI(N)}{= \calD_F}{}
\addConstraint{\bu_k}{\in \hat{\calU}_{\calI(k)} \quad \label{eq:input_constraints_hl_mpc}}{\forall k \in [0, \: N-1]}
\addConstraint{\bx_N}{\in \calE \label{eq:terminal_constraints_hl_mpc}}
\end{mini!}
where $J^H$ is the cost function, $\bar{x}^j$ gives the initial state of the system, and $\calI(k)$ gives the domains. The resulting trajectory and input given by the solution are $\bx^j$ and $\bu^j$. Denote the time over the MPC horizon by $t_N$. Note that there are at most $r_N = N/2$ reset maps (from Assumption \ref{assump:reset_not_on_guard}). The traversal over the hybrid graph is guaranteed to respect the directed graph because all the transitions are determined via the action of a reset map. The dynamics constraint enforces the continuous dynamics for the mode, and in the case that the state is on a guard, the associated reset map is applied. The integration \eqref{eq:discrete_int} terminates early if the state hits the guard, thus preventing the state from ``skipping" over the guard.

\textbf{Fixed mode MPC}. This layer takes the sequence of modes and guards to determine a trajectory and control action. By fixing the modes and guards, the problem becomes a time-varying and state-constrained problem. One of the benefits of breaking the solution into layers is that this MPC problem is in general much quicker to solve than the high level MPC problem because no discrete decisions need to be made. Define the dynamics as
\begin{equation}
    \label{eq:dynamics_fixed_modes}
    g_{\calI, G}(\bar{x}, \bar{u}, k) = \begin{cases}
        \bar{f}_{\calI(k)}(\bar{x}, \bar{u}) &\text{if} \; \calI(k) = \calI(k+1) \\
        \calR_{G(k)}(\bar{x}) &\text{if} \; \calI(k) \neq \calI(k+1)
    \end{cases}
\end{equation}
where $\calI(k)$ gives a time dependent mode sequence and $G(k)$ gives the active guard at time $k$. Note that the integration to get the discrete time system does not have a guard termination condition. There is no integration or time passage associated with the dynamics in the second case. Given fixed guards, the states needs to avoid the other guards and contact the desired guard. We denote these state constraints as follows. Let $P_\calI(k) := \set{p}{p_{(\calI(k), \calM)} \in \calG \; \forall (\calI(k), \calM) \in \Gamma}$ denote the set of all possible guards from the domain $\calI(k)$. Then define the constraint set as:
\begin{equation*}
    \calX_{\calI, G}(k) = \begin{cases}
        \set{x}{h_{G(k)}(x) = c_{G(k)}} \; \text{if} \; \calI(k) \neq \calI(k + 1) \\
        \set{x}{h_{p}(x) < c_{p} \; \forall p \in P_\calI(k)} \; \text{otherwise}
    \end{cases} 
\end{equation*}
Without loss of generality, we used $h_{p}(x) < c_{p}$ to indicate that the guard condition is not met, although in practice it is possible that the condition is $h_{p}(x) > c_{p}$. 

This MPC problem at time step $j$ is given in \eqref{eq:ll_mpc}.
\begin{mini!}[1]
{\bx, \; \bu}{J^L(\bx, \bu)}
{\label{eq:ll_mpc}}{}
\addConstraint{\bx_0}{= \bar{x}^j \label{eq:C2_ll_mpc}}
\addConstraint{\bx_{k+1}}{= g_{\calI, G}(\bx_k, \bu_k, k) \quad \label{eq:dyn_ll_mpc}}{\forall k \in [0, \: N]}
\addConstraint{\bu_k}{\in \hat{\calU}_{\calI(k)} \quad \label{eq:C3_ll_mpc}}{\forall k \in [0, \: N-1]}
\addConstraint{\bx_k}{\in \calX_{\calI, G}(k) \quad \label{eq:state_constraints_ll_mpc}}{\forall k \in [1, \: N]}
\addConstraint{\bx_N}{\in \calE \label{eq:C5_ll_mpc}}
\end{mini!}
\begin{remark}
    Since we are only considering the discrete time system, we make the standing assumption throughout the paper that if the discrete state satisfies an inequality constraint, then the continuous time system satisfies that constraint between discrete nodes.
\end{remark}
When this MPC problem is run in closed loop, at the next time instant, the previous time node is removed from its domain and is replaced with a new node in the final domain. This shifts the domain schedule so that after a domain is exited, all the constraints related to that domain are removed and after $N$ time steps, we get $\calI(k) = \calD_F \; \forall k$.

\textbf{Low level controller}. This controller is designed to track the trajectory generated by the low level MPC. Although optional, in practice this layer can be very helpful, especially for unstable systems.

\subsection{Stability}
We present a definition for stability of hybrid systems.
\begin{definition} [Hybrid Stability]
\label{def:hybrid_stability}
    The discretized version of the system, denoted as $\bar{\calH}$, is \emph{Hybrid Stable} in a region $\calF \subset \bbR^n$ if it satisfies two conditions:
    \begin{equation}
    \label{eq:hybrid_asmpy_stability}
        \bar{x} \in \calF \Rightarrow \lim_{j \rightarrow \infty} \norm{\bar{x}^j} = 0.
    \end{equation}
    and secondly, there exists some $\calA \subset \calF$ in which the system is locally stable in the sense of Lyapunov:
    \begin{equation}
    \label{eq:hybrid_lyap_stability}
        \forall \epsilon > 0, \; \exists \delta > 0 \quad \text{s.t} \quad \norm{\bar{x}^j} < \delta \Rightarrow \norm{\bar{x}^i} < \epsilon \; \forall i \geq j.
    \end{equation}
\end{definition}
\begin{assumption} [Properties of the Equilibrium]
\label{assump:properties_of_eq}
The terminal constraint set does not include a guard: there exists some $\epsilon \geq c_f$ such that $\levset{\epsilon}{\Phi} \cap \calG = \emptyset$. Further, $\bar{f}(0,0) = 0$.
\end{assumption}

Here we introduce the suboptimal control strategy used for the high level MPC problem. 
The input policy at the $j$th time step is $\bu^j = Y(\bar{x}^j, \bu^{j-1}) := (\bu^{j-1}_1, ..., \bu^{j-1}_{N-1}, \kappa_f(\bx^j_N))$. The initial policy is any feasible input sequence, which is in general non-unique. We denote the set of inputs that respect the constraints in \eqref{eq:hl_mpc} as $U(\bar{x})$ where $\bar{x}$ gives the current initial state and by extension the current domain. $\bu^j$ is a function of $\bu^{j-1}$ and, therefore, we will investigate the extended state $\bar{z}^j := (\bar{x}^j, \bu^j)$, which satisfies a difference inclusion. For more details see \cite{allan_inherent_2017}. We denote this difference inclusion as:
\begin{multline}
\label{eq:extended_diff_inclusion}
    \bar{z}^{j+1} \in H(z^j) := \set{(\bar{x}^{j+1}, \bu)}{ \\ \bar{x}^{j+1} = \bar{f}_{\calI(\bar{x}^j)}(\bar{x}^j, \hat{\bu}_1) \; \text{if} \; \bar{x}^j \notin \calG, \\ \bar{x}^{j+1} = R_{G(\bar{x}^j)}(\bar{x}^j) \; \text{if} \; \bar{x}^j \in \calG, \: \bu = Y(\bar{x}^j, \hat{\bu}), \: \hat{\bu} \in U(\bar{x}^j)}
\end{multline}
where $\calI(\bar{x}^j)$ gives the domain of the current state and $G(\bar{x}^j)$ gives the guard at the current state. Given an extended state, we denote the resulting state trajectory as $\bx(z)$ (i.e. generated by application of the inputs to the dynamics). To simplify notation we write $J^H(\bar{z}^j) := J^H(\bx(\bar{z}^j), \bu^j)$ (where $\bu^j$ comes from $\bar{z}^j$)\footnote{Note that $J^H$ is not continuous as a function of $\bar{z}^j$}.
\begin{lemma}
\label{lemma:bounded_cost}
    Under Assumptions \ref{assump:cont_cost}  and \ref{assump:constraint_sets} the cost $J^H$ is bounded for all $\bar{z}$ given by the difference inclusion if the initial state in $\bar{z}$ satisfies $\bar{x}^0 \in \calF \subset \bbR^n$ for any compact $\calF$.
\end{lemma}
\begin{proof}
    Let $\tilde{L}_u$ and $\tilde{L}_x$ denote the largest input and state Lipschitz constants for all $f_\calI$. We now bound the state over the MPC horizon $t_N$: from Gronwall Bellman \cite{khalil_nonlinear_2002} we get $\norm{\bar{x}^0 - \bx^0_k} \leq \tilde{L}_u M_u t_N e^{\tilde{L}_x t_N} \; \forall k \leq N$ where $M_u := \max_{\calI \in \calJ}\{\sup_{u \in \calU_\calI} \norm{u}\} < \infty$. Thus, there is some constant $c_1$ such that $\norm{\bx_k^0} \leq c_1 \; \forall k \leq N$. Therefore, the state at the first guard impact, $\bx^0_{k_{r1}}$, is in a compact set and thus $\calR(\bx^0_{k_{r1}})$ is bounded. This argument can be repeated for each domain given that the result of the reset map is bounded and there are at most $r_N$ resets. We denote the largest upper bound as $c_2$. In turn we have that there is a compact set $B := \set{x}{\norm{x} \leq c_2}$ such that $\bx(\bar{z}) \subset B^N$. From Assumption \ref{assump:cont_cost} we have that $J^H(\bx, \bu)$ is continuous on $\bbR^{n \times N} \times \bbR^{\hat{m} \times N}$ and thus, bounded on the compact set $B^N \times \hat{\calU}^N$. Therefore, $J^H(z)$ is bounded for any $\bar{x}^0 \in \calF$.
\end{proof}
Let $\Pi_x(\levset{\rho}{J^{(\cdot)}}) := \set{x}{J^{(\cdot)}(z) \leq \rho, \; z = (x, \bu)}$ denote the projection of the sub-level set onto the state variables.
\begin{theorem} [High Level Hybrid Stability]
\label{thm:hl_mpc_stability}
    Let $\calF \subset \bbR^n$ denote a compact set of states where the problem \eqref{eq:hl_mpc} is feasible. Then if Assumptions \ref{assump:cont_cost} - \ref{assump:properties_of_eq} hold, there exists a controller that stabilizes $\bar{\calH}$ in the sense of Def. \ref{def:hybrid_stability} for all $\bar{x}^0 \in \calF$.
\end{theorem}
\begin{proof}
    The proof takes two parts: first, we show that the cost converges to some value $c > 0$ in finite time, then we show that within this sub-level set the cost function acts as a Lyapunov function for the system.
    
    \textit{Finite time convergence}. From Assumptions \ref{assump:terminal_stability}, \ref{assump:stage_cost_bound}, Eq. \eqref{eq:extended_diff_inclusion}, and the proof of Theorem 14 in \cite{allan_inherent_2017} we have that 
    $J^H(z^{j+1}) \leq J^H(z^{j}) - \alpha_3(\norm{z^j})$ for some $\alpha_3 \in \calK$. Note that $\alpha_3(\norm{z}) \geq c_3$ for some $c_3 > 0$ for all $z \notin \levset{c}{J^H}$. Therefore $J^H(z^{j+1}) \leq J^H(z^j) - c_3$. Using this relation recursively starting from $j = 0$ yields $J^H(z^{j+1}) \leq J^H(z^0) - (j+1)c_3$. By Lemma \ref{lemma:bounded_cost} $J^H(z^0) = c_4 < \infty$. Thus $J^H(z^{j+1}) \leq c_4 - jc_3$ for any $j$ while $z^j \notin \levset{c}{J^H}$. One can see that there exists a $j = N_c$ with $0 < N_c < \infty$ such that $c_4 - N_c c_3 \leq c$ and therefore after $N_c$ steps $z^{N_c} \in \levset{c}{J^H}$.

    \textit{Lyapunov Stability}. Let $c > 0$ be any number such that $\Pi_x(\levset{c}{J^H}) \subset \levset{\epsilon}{\Phi}$ where $\epsilon$ is defined in Assumption \ref{assump:properties_of_eq}. This will always exist under the conditions in Assumption \ref{assump:cont_cost}. Within $\Pi_x(\levset{c}{J^H})$, the dynamics are continuous, there are no guards, and the set is simply connected; therefore, we can apply the results from \cite{allan_inherent_2017} Theorem 14 which yield that the system in this region is asymptotically stable in the sense of Lyapunov.

    Therefore for all $\bar{x}^0 \in \calF$ the system is asymptotically stable and in $\Pi_x(\levset{c}{J^H})$ the system is also stable in the sense of Lyapunov.
\end{proof}

\begin{proposition}
    Let Assumptions \ref{assump:cont_cost} - \ref{assump:properties_of_eq} hold. Let the set $\calF \subset \bbR^n$ denote a compact set of states where the problem \eqref{eq:hl_mpc} is feasible. For all $\bar{x}^0 \in \calF$, given $\calI(k), G(k)$ from the (potentially suboptimal) solution of \eqref{eq:hl_mpc}, the low level MPC given by \eqref{eq:ll_mpc} is feasible and there exists a stabilizing controller in the sense of Def. \ref{def:hybrid_stability}. 
\end{proposition}
\begin{proof}
    The solution of \eqref{eq:hl_mpc} provides a feasible input and state trajectory that satisfy the constraints by construction. From Theorem \ref{thm:hl_mpc_stability} there is a stabilizing controller.
\end{proof}
 
\subsection{Robustness}
Given disturbances in the system, the actual state will deviate from the planned state by some amount. We will quantify this deviation through the use of a tube.

Define set addition between $A$ and $B$ as $A \oplus B := \set{a + b}{ a \in A, \; b \in B}$. 

\begin{definition} [Tube and Tube Diameter]
\label{def:tube}
    Given a trajectory $\tilde{x}(t)$, define a tube as $\calT_{\tilde{x}}(t) := \calT(t) \oplus \{\tilde{x}(t)\}$ where $\calT(t)$ is a time varying cross section. This cross section is defined as some set with $0 \in \calT(t)$ for all $t$. Further, let $\diam{\calT_{\tilde{x}}(t)}:= \sup_{x \in \calT(t)} \norm{x}$ which satisfies the following conditions: (a) $\diam{\calT_{\tilde{x}}(0)} = 0$, (b) $\diam{\calT_{\tilde{x}}(t)}$ is non-decreasing.
 A tube is valid for a system if $x(\tau) \in \calT_{\tilde{x}}(\tau) \; \forall \tau \in [t_i, t_i + T]$ for some $T > t_i$.
\end{definition}

\begin{proposition} [Tube Existence]
    \label{prop:tube_exists}
    Consider the system in \eqref{eq:continuous_sys} for a single domain and let Assumption \ref{assump:bounded_disturbance} hold. Given a dynamically feasible trajectory (for the undisturbed system), $\tilde{x}(t)$, and associated feed forward inputs, $\tilde{u}(t)$, with initial condition $\tilde{x}(t_0) = x(t_0)$, if the input to the system satisfies $\norm{\tilde{u}(t) - u(t)} \leq \zeta$ for some $0 \leq \zeta < \infty $ then a time varying tube, $\calT_{\tilde{x}}(t)$, always exists satisfying the conditions in Definition \ref{def:tube}. We refer to this tube as the trivial tube.
\end{proposition}
\begin{proof}
    \begin{align*}
        \norm{x(t) - \tilde{x}(t)} &\leq \int_{t_0}^{t} \norm{f(x, u) + w - f(\tilde{x}, \tilde{u})} ds \quad \forall t \geq t_0 \\
        &\leq \int_{t_0}^{t} L_u\norm{u - \tilde{u}} + L_x\norm{x - \tilde{x}} ds + \eta (t - t_0) \\
        &\leq (\eta + L_u\zeta)(t - t_0)e^{L_x (t - t_0)}
    \end{align*}
    where $L_u$ and $L_x$ are the Lipschitz constant of the system for the input and the state respectively. We used the fact that $x(t_0) = \tilde{x}(t_0)$, and the last step uses the Gronwall-Bellman Inequality \cite{khalil_nonlinear_2002}. Therefore, we can construct a tube of circular cross section, specifically, $\calT(t) = \set{x}{\norm{x} \leq (\eta + L_u\zeta)(t - t_0)e^{L_x (t - t_0)}}$. We can trivially see that for $t = t_0$ $\calT(t) = \{0\}$ and $\diam{\calT(t)} = (\eta + L_u\zeta)(t - t_0)e^{L_x (t - t_0)}$ which is non-decreasing.
\end{proof}
\begin{remark}
    Note if we only apply the feed forward input then $\norm{x(t) - \tilde{x}(t)} \leq \eta(t - t_i)e^{L_x (t - t_i)}$ and we can recover a tube even with no low level controller.
\end{remark}

As time increases, the tube's diameter may increase to infinity. We only require that the tube has finite diameter over some finite time. Therefore, we cannot plan according to classical tube MPC.

Many methods for certifying tubes require the system to be feedback linearizable (e.g. \cite{csomay-shanklin_multi-rate_2022, lopez_dynamic_2019}) which is not the case for many real world systems. Further, it is possible that a controller could, in theory, certify a tube, but the tube size is actually larger than the region of attraction of the controller. If a local linearization is used to create a LQR controller, it could stabilize the system, but the tube size may be larger than the region of attraction of the LQR controller.

Here we quantify the transient benefits of adding a controller that can provide a E-ISS-Lyapunov function certificate to the error tracking system, but where its region of attraction may be smaller than the tube that can be certified. Consider the error about a continuous time trajectory of length $T$:
\begin{align*}
    x_e(t) &:= x(t) - \tilde{x}(t) \\
    u_e(t) &:= u(t) - \tilde{u}(t) \\
    \dot{x}_e &= f_e(x_e, u_e) + w(t)
\end{align*}

\begin{proposition}
\label{prop:small_tube}
    Let $\calT_{\tilde{x}}(t)$ denote the tube generated by Proposition \ref{prop:tube_exists}. Let $R_e$ denote the region of attraction of the controller in the error system. Assume that there exists some $\epsilon > 0$ such that $\set{x_e}{\norm{x_e} \leq \epsilon} \subset R_e$. If the low level controller certifies an E-ISS Lyapunov function, $\calV_e(x_e(t))$, for the undisturbed error system, Assumption \ref{assump:bounded_disturbance} holds, $\tilde{x}(t_0) = x(t_0)$, and $\sigma(\eta)/k_1 \leq \eta$, then there exists a $\tau > 0$ such that
    \begin{equation*}
        x(t) \in \calT'_{\tilde{x}}(t) \quad \forall t \in [t_0, t_0 + \tau]
    \end{equation*}
    where $\diam{\calT'_{\tilde{x}}(t)} \leq \diam{\calT_{\tilde{x}}(t)}$ for all $t \in [t_0, t_0 + \tau]$. If $\tau < T$ then for $t \in [t_0 + \tau, t_0 + T]$ the state satisfies $x(t) \in \calT_{\tilde{x}}(t)$. Further, $\calT'_{\tilde{x}}(t)$ satisfies all the properties in Definition \ref{def:tube}.
\end{proposition}
\begin{proof}
    Due to the E-ISS property of the low level controller (Definition \ref{def:iss_lyap}) on the error system
    \begin{align*}
        \dot{\calV}_e(x_e) &\leq -k_3\norm{x_e} + \sigma(\eta) \\
        &\leq -k_3 k_2^{-1} \calV_e(x_e) + \sigma(\eta).
    \end{align*}
    Let $\gamma := k_3 k_2^{-1}$. By the comparison lemma \cite{khalil_nonlinear_2002} it follows that
    \begin{multline*}
        \calV_e(x_e(t)) \leq \calV_e(x_e(t_0))e^{-\gamma (t-t_0)} \\+ \sigma(\eta)\gamma^{-1}(1 - e^{-\gamma (t - t_0)}).
    \end{multline*}
    Then by Definition \ref{def:iss_lyap} we get
    \begin{multline*}
        \norm{x_e(t)} \leq k_2k_1^{-1}\norm{x_e(t_0)}e^{-\gamma (t-t_0)} \\+ \sigma(\eta)\gamma^{-1}k_1^{-1}(1 - e^{-\gamma (t - t_0)}).
    \end{multline*}
    By assumption $x_e(t_0) = 0$, leading to
    \begin{equation*}
        \norm{x_e(t)} \leq \sigma(\eta)\gamma^{-1}k_1^{-1}(1 - e^{-\gamma (t - t_0)}).
    \end{equation*}
    Here we compare to the bound in Proposition \ref{prop:tube_exists} when just the feed forward input is applied. Therefore we want to give conditions under which
    \begin{equation*}
        \underbrace{\sigma(\eta)\gamma^{-1}k_1^{-1}(1 - e^{-\gamma (t - t_0)})}_{a(t)} \leq \underbrace{\eta (t - t_0) e^{L_x (t - t_0)}}_{b(t)}.
    \end{equation*}
     We will show the inequality via the comparison lemma. Taking derivatives,
    \begin{align*}
        \dot{a}(t) &= -\gamma a(t) + \sigma(\eta)/k_1 \\
        \dot{b}(t) &= L_x b(t) + \eta e^{L_x (t - t_0)}.
    \end{align*}
    By their definition $a(t_0) = b(t_0) = 0$ and therefore $\dot{a}(t_0) \leq \dot{b}(t_0)$ if $\sigma(\eta)/k_1 \leq \eta$. Then note that $\dot{a}(t) \leq \dot{a}(t_0)$ and $\dot{b}(t) \geq \dot{b}(t_0)$ for all $t > t_0$ because $\gamma > 0$, $a(t) > 0$, and $b(t) > 0$. Thus by the comparison lemma $a(t) \leq b(t) \; \forall t > t_0$. We can thus get the cross section $\calT'(t) = \set{x}{\norm{x} \leq a(t)}$ with $\diam{\calT'_{\tilde{x}}(t)} = a(t)$. This trivially satisfies the conditions in Definition \ref{def:tube}.

    Define $\tau$ as follows
    \begin{argmini!}
        {t \geq 0}{t \label{eq:tr_objective}}
        {\label{eq:tr_opt}}{\tau = }
        \addConstraint{\set{x_e}{\norm{x_e} \leq a(t_R)}}{\not \subseteq R_e \label{eq:C1_tr_opt}}{},
    \end{argmini!}
    where $\tau = \infty$ if the problem is infeasible. If the set defined by $a(t)$ is completely contained in the region of attraction of the controller, then \eqref{eq:tr_opt} will be infeasible and thus $\tau = \infty$ which certifies a finite sized positive invariant set. Otherwise, under the assumption on the region of attraction, continuity of $a(t)$, and because $a(t_0) = 0$, it follows that there is some $t_1 > 0$ such that $\set{x_e}{\norm{x_e} \leq a(t)} \subseteq R_e \; \forall t \leq t_1$. This implies that $\infty > \tau > t_1 > 0$. At $\tau$ the state will leave the region of attraction of the controller. At this point the controller can no longer provide additional stabilizing action, and thus the tube generated by $a(t)$ is no longer valid. In this case, after the region of attraction is lost, we can bound the state through the bound in proposition \ref{prop:tube_exists}.  
\end{proof}
\begin{remark}
    This highlights that even if we cannot certify a finite diameter tube, adding a controller can shrink the tube size over some time.
\end{remark}

Given these tube results, so long as MPC is re-planned from the current state of the system, then the state error can never be larger than $\diam{\calT(\delta t)}$ where $\delta t$ is the computation period of MPC. Therefore we will use $\diam{\calT(\delta t)}$ as the error tube for planning. To stay with MPC theory, we assume that the integration step matches the MPC computation period: $\delta t_d = \delta t$, although in practice it is common to let $\delta t \leq \delta t_d$. 

\begin{definition} [Tube Based Robust Asymptotic Stability]
    The origin is robustly asymptotically stable in $\calF \subset \bbR^n$ if there exists some $\delta > 0$ such that for all $\diam{\calT(\delta t)} \leq \delta$ we have that both $\calF$ is robustly positive invariant and there exists a $\beta \in \calK \calL$ and $\gamma \in \calK$ such that for each $\bar{x}^0 \in \calF$ we have that
    \begin{equation*}
        \norm{\bar{x}^j} \leq \beta(\norm{\bar{x}^0}, j) + \gamma(\diam{\calT(\delta t)}).
    \end{equation*}
\end{definition}
The following proposition is a tube version of Theorem 21 in \cite{allan_inherent_2017} modified for the single domain hybrid case.
\begin{proposition} [Inherent Robustness for Single Domain MPC]
\label{prop:nominal_robustness_final}
        Let Assumptions \ref{assump:cont_cost} - \ref{assump:properties_of_eq} hold. Denote the nominal region of attraction of \eqref{eq:ll_mpc} in $\calD_F$ as $\calF_F$. Let $\calI(k) = \calD_F \; \forall k$. For all $\rho > 0$ such that $\Pi_x(\levset{\rho}{J^L}) \subseteq \ \calX_{\calI, G}(k) \subset \calD_F$ there exists some $\delta > 0$ such that if $\diam{\calT(\delta t)} \leq \delta$ then the equilibrium is robustly asymptotically stable in the set $\calS := \Pi_x(\levset{\rho}{J^L}) \cap \calF_F$ under the closed loop controller defined by \eqref{eq:ll_mpc}.
\end{proposition}
\begin{proof}
    By construction, $\bar{f}$ is continuous in $\calS$ and $\bar{x}^j \in \calX_{\calI, G} \; \forall k$, so all the conditions of Theorem 21 in \cite{allan_inherent_2017} are satisfied. We leverage this result, but rather than using a disturbance to bound the state error, we use the tube diameter: $\norm{\bar{x}^{j+1} - \bx^j_1} \leq \diam{\calT_{\bx^j_1}(\delta t)}$. Then by the results of Theorem 21 in \cite{allan_inherent_2017}, the origin is robustly asymptotically stable in $\calS$ if no guard conditions are met. By construction, $\calS \cap \calG = \emptyset$ thus rendering the set robustly asymptotically stable under the controller \eqref{eq:ll_mpc}. 
\end{proof}
\begin{remark}
    This tells us that we get a ISS type condition for the system where the disturbance is quantified through the tube size. In general decreasing $\delta t$ will cause the tube to be smaller, and thus the positively invariant set will be smaller.
\end{remark}
To control this system through multiple domains in the face of disturbances, we must create a robust fixed mode MPC. There are two key changes introduced: first, we tighten the guard-avoid state constraints as dictated by the tube, and, secondly, an additional ``virtual" pass through node is introduced to force the true state to trigger the desired guard conditions. This virtual node is added into each domain where a guard condition is met. The node will be constrained to move ``past" the guard such that the entire error tube passes through the guard as depicted in Fig. \ref{fig:virtual_nodes}. The last time node in the $n$th domain is denoted as $k_n$. We define the dynamics used in this version of MPC as follows
\begin{equation*}
    \label{eq:dynamics_fixed_modes_robust}
    \tilde{g}_{\calI, G}(\bar{x}, \bar{u}, k) = 
    \begin{cases}
         \bar{f}_{\calI(k)}(\bar{x}, \bar{u}) &\text{if} \; \calI(k) = \calI(k+1) \\
         \calR_{G(k)}(x(t_g)) &\text{if} \; \calI(k) \neq \calI(k + 1) \\
    \end{cases}
\end{equation*}
where $x(t_g)$ denotes the state at which the planned continuous time system impacts the guard. At a given solve time, $j$, $x(t_g)$ is defined completely by $\bar{x}^j_{k_n - 1}$, $\bar{x}^j_{k_n}$, and $\bar{u}^j_{k_n}$ which just makes the second case a normal nonlinear constraint.
\begin{remark}
    Here we implicitly assume that the domain and dynamics are defined beyond the guard. Due to the guard impact constraint, there is a time when $\bar{x}^j \in \calT_{\bx^{j-1}_1}(\delta t)$ with $\bx^{j-1}_1 \in G(k_n)$ for each planned guard. Thus we further assume that for any $\bar{x} \in \calT_{x \in G(k_n)}(\delta t)$ that $\set{\bar{u}}{\bar{x}^+ = \bar{f}(\bar{x}, \bar{u}), \: \bar{u} \in \hat{\calU}_{\calI(k)}, \: \calT_{\bar{x}^+}(\delta t) \subset  \set{x}{h_{G(k_n)}(x) \geq c_{G(k_n)}}} \neq \emptyset$. This is true for many robotic systems as the guards are generally physical surfaces and the dynamics are defined in all of $\bbR^n$.
\end{remark}

\begin{figure}
    \vspace{2mm}
    \centering
    \includegraphics[width=0.6\linewidth]{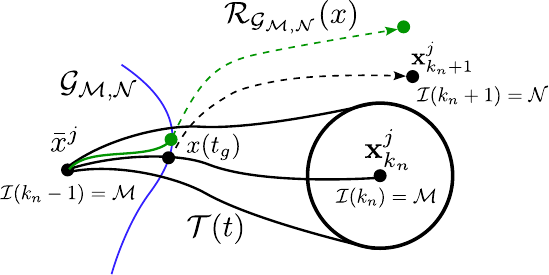}
    \vspace{-3mm}
    \caption{Cartoon depiction of the additional node. The blue line depicts the guard, the solid green line depicts the flow of the system, and the dotted lines show the action of the reset map. Given that the tube at the next node is completely past the guard, at some point the true state will hit the guard.}
    \vspace{-4mm}
    \label{fig:virtual_nodes}
\end{figure}
We shrink the guard-avoid constraints as follows:
\begin{equation*}
    \tilde{\calX}_{\calI, G}(k) = \begin{cases}
        \set{x}{h_{G(k)}(x) = c_{G(k)}} \; \text{if} \; k = k_n - 1 \\
        \set{x}{h_{p}(y) < c_{p} \forall p \in P_\calI(k), \forall y \in \calT_x(\delta t)} \: \text{o.w.}
    \end{cases} 
\end{equation*}
We must also shrink the input constraints by the maximum amount of input used by the low level controller. We denote this shrunk set by $\tilde{\calU}_{\calI} \subset \bbR^{\hat{m}}$. 

The robust fixed mode MPC problem at time step $j$ is
\begin{mini!}[1]
{\bx, \; \bu}{J^L(\bx, \bu)}
{\label{eq:robust_ll_mpc}}{}
\addConstraint{\bx_0}{= \bar{x}^j \label{eq:C2_ll_mpc}}
\addConstraint{\bx_{k+1}}{= \tilde{g}_{\calI, G}(\bx_k, \bu_k, k) \quad \label{eq:dyn_ll_mpc}}{\forall k \in [0, \: N]}
\addConstraint{\bu_k}{\in \tilde{\calU}_{\calI}(k) \quad \label{eq:C3_ll_mpc}}{\forall k \in [0, \: N-1]}
\addConstraint{\bx_k}{\in \tilde{\calX}_{\calI, G}(k) \quad \label{eq:C3_ll_mpc}}{\forall k \in [1, \: N]}
\addConstraint{\calT_{\bx_{k_n}}(\delta t)}{\subset \set{x}{h_{G(k)}(x) \geq c_{G(k)}} \label{eq:robust_ll_mpc_tube_guard}} {\forall k_n}
\addConstraint{\bx_N}{\in \calE \label{eq:C5_ll_mpc}}
\end{mini!}
Note that it is possible to re-formulate the hybrid problem with the same shrunk guard-avoid constraints so that feasibility of this robust hybrid problem implies feasibility of the robust fixed mode problem.

When \eqref{eq:robust_ll_mpc} is solved in the loop, we need to change how the time nodes are updated relative to the strategy for \eqref{eq:ll_mpc}. The previous node from the current domain is still removed and an additional node is added to the final domain, but we must take special care when interacting with the guard. When the guard is impacted by the true continuous time system state, we immediately remove all remaining nodes in that domain, including the virtual node, adding them to the final domain, and re-compute the MPC solution from the new state.

In the sequel, we will be concerned with the benefits of the modified constraints for the fixed mode MPC when the optimization is feasible, and therefore we choose to work under the following assumption.

\begin{assumption}
\label{assump:feasible_rpi}
    Let $\tilde{\calF} \subset \bbR^n$ be the set of states where \eqref{eq:robust_ll_mpc} is feasible. There exists some $\epsilon > 0$ such that for all $\diam{\calT(\delta t)} \leq \epsilon$, there is a set $\hat{\calF} \subset \tilde{\calF}$ that is robustly positive invariant under the closed loop MPC.
\end{assumption}
\begin{remark}
    Although this may sound like a strong assumption, we find in practice for robotic systems that this assumption is often true. Further, in situations where the fixed mode MPC is not feasible, if the robust version of the hybrid problem is feasible, then re-solving the hybrid MPC will put the current state in a feasible set. In fact, this is the only reason to use the hybrid optimal control problem in a MPC fashion over just selecting a fixed mode sequence once\footnote{Theoretically, better optimality could be achieved, but as discussed above, optimality is often seen as a secondary objective in robotic systems. Secondly, it can be very difficult to converge to a global optimum given the non-convex and combinatorial nature of the problem.}. Every layer in the stack provides some form of robustness. Given this, there are a number of tools presented in this paper to reduce the effect of disturbances and thus keep this problem feasible: use of a low level controller, quicker MPC computation periods, and re-computation of the hybrid problem as often as possible.
\end{remark}

\begin{theorem}
\label{thm:ll_robust_mpc}
    Consider $\calI, G$ generated by a feasible solution to Eq. \eqref{eq:hl_mpc} and modified for additional robustness as discussed above. Let Assumptions \ref{assump:cont_cost}-\ref{assump:feasible_rpi} hold. If the image of the last reset map is a subset of $\calS$, and $\calT(\delta t)$ is valid for any disturbances $w \in \calW$ then, for any $x^0 \in \hat{\calF}$, the state will converge to $\calS$ in finite time and then become robustly asymptotically stable to the equilibrium.
\end{theorem}
\begin{proof}
    Given the state $\bar{x}^j$, the trajectory satisfies $\bx^j_1 \in \tilde{\calX}_{\calI, G}(1)$, and therefore $\bar{x}^{j+1}$ is within the guard-avoid constraints given by $\calX_{\calI, G}(1)$. Recursively applying this argument yields that the guard-avoid state constraints are robustly satisfied under Assumption \ref{assump:feasible_rpi}.
    
    Under the in-the-loop node modification strategy proposed above, we have that for each domain with a guard impact constraint there exists a time step where $\bar{x}^j \in \calT_{\bx^{j-1}_{k_n - 1}}(\delta t)$ with $\bx^{j-1}_{k_n - 1} \in G(k_n)$. If the guard constraint has been satisfied, then the remaining nodes are moved to the final domain and the MPC is re-planned. Otherwise, the only node left in this domain beyond the initial condition is the $k_n$th node, which is the virtual node. Denote the flow of the continuous time system from this state at time $t$ as $\varphi(\bar{x}^j, t)$, which is continuous in this domain. Because $\bar{x}^j \notin G(k_n)$, $h_{G(k_n)}(\varphi(\bar{x}^j, 0)) < c_{G(k_n)}$. Ignoring the effects of the reset map, there exists a time, $t_2 \leq \delta t$ when $h_{G(k_n)}(\varphi(x^j, t_2)) \geq c_{G(k_n)}$ from the tube definition and $\eqref{eq:robust_ll_mpc_tube_guard}$. By the intermediate value theorem, there is a time, $t_3 \leq t_2$ when $h_{G(k)}(\varphi(\bar{x}^j, t_3)) = c_{G(k_n)}$. Since the reset map happens only exactly on the guard, it is safe to ignore the effects for the sake of this argument. When the state hits the guard, the reset map takes action, and at that point, the condition has been satisfied.

    From this, the state will traverse the domain sequence under the effects of disturbances, and therefore, the state arrives in $\calD_F$ in finite time. By assumption, this state is in $\calS$. Then the results of Proposition \ref{prop:nominal_robustness_final} apply.
\end{proof}
\section{Simulation Results}
To demonstrate the layered MPC approach we consider its application to a bouncing ball and walking robot.

\textbf{Bouncing ball.}  The first example considers a ball with nonlinear drag, a horizontal wall, a vertical wall, and a circular guard in the first quadrant. There are two domains: inside and outside of the circle. Denote the state as $s = [x, y, \dot{x}, \dot{y}]^T$. When outside the circle, the dynamics are:
\begin{align*}
    \begin{bmatrix}
        \ddot{x} \\ \ddot{y}
    \end{bmatrix} = \begin{bmatrix}
        -\gamma \norm{(\dot{x}, \dot{y})}^2 \dot{x} + \min\{u_x, 0\}/m \\
        -\gamma \norm{(\dot{x}, \dot{y})}^2 \dot{y} + \min\{u_y, 0\}/m + g \\
    \end{bmatrix} 
\end{align*}
and inside the circle the dynamics become:
\begin{align*}
     \begin{bmatrix}
        \ddot{x} \\ \ddot{y}
    \end{bmatrix} = \begin{bmatrix}
        -\gamma \norm{(\dot{x}, \dot{y})}^2 \dot{x} + u_x/m \\
        -\gamma \norm{(\dot{x}, \dot{y})}^2 \dot{y} + u_y/m \\
    \end{bmatrix}.
\end{align*}
where $\gamma$ is the drag coefficient, $g$ is gravity, and $m$ is the mass. The position and velocities are normal integrators.
The reset maps for the walls are:
\begin{align*}
    \calR_x(s) = [x, y, \dot{x}, -\dot{y}]^T, \quad \calR_y(s) = [x, y, -\dot{x}, \dot{y}]^T.
\end{align*}
The reset map for the circular guard is the identity map.

The goal is to stabilize to the middle of the circle. Disturbance forces happen in all directions. Fig. \ref{fig:bouncing_ball_disturbances} demonstrates how the control architecture can be leveraged to robustly stabilize trajectories. For the simulation, $w \in \calW = \set{(w_x, w_y)}{|w_x| \leq 10, \; |w_y| \leq 10}$, $M = 1$, $\gamma = -0.02$ and $u_i \in [-100, 100] \; i \in \{x, y\}$. Casadi \cite{andersson_casadi_2019} and IPOPT \cite{wachter_implementation_2006} are used to formulate and solve the fixed mode MPC problem. 
The cross entropy method \cite{rubinstein_cross-entropy_nodate} is used to solve the hybrid planning problem.

\begin{figure}
    \vspace{2mm}
    \centering
    \includegraphics[width=1\linewidth]{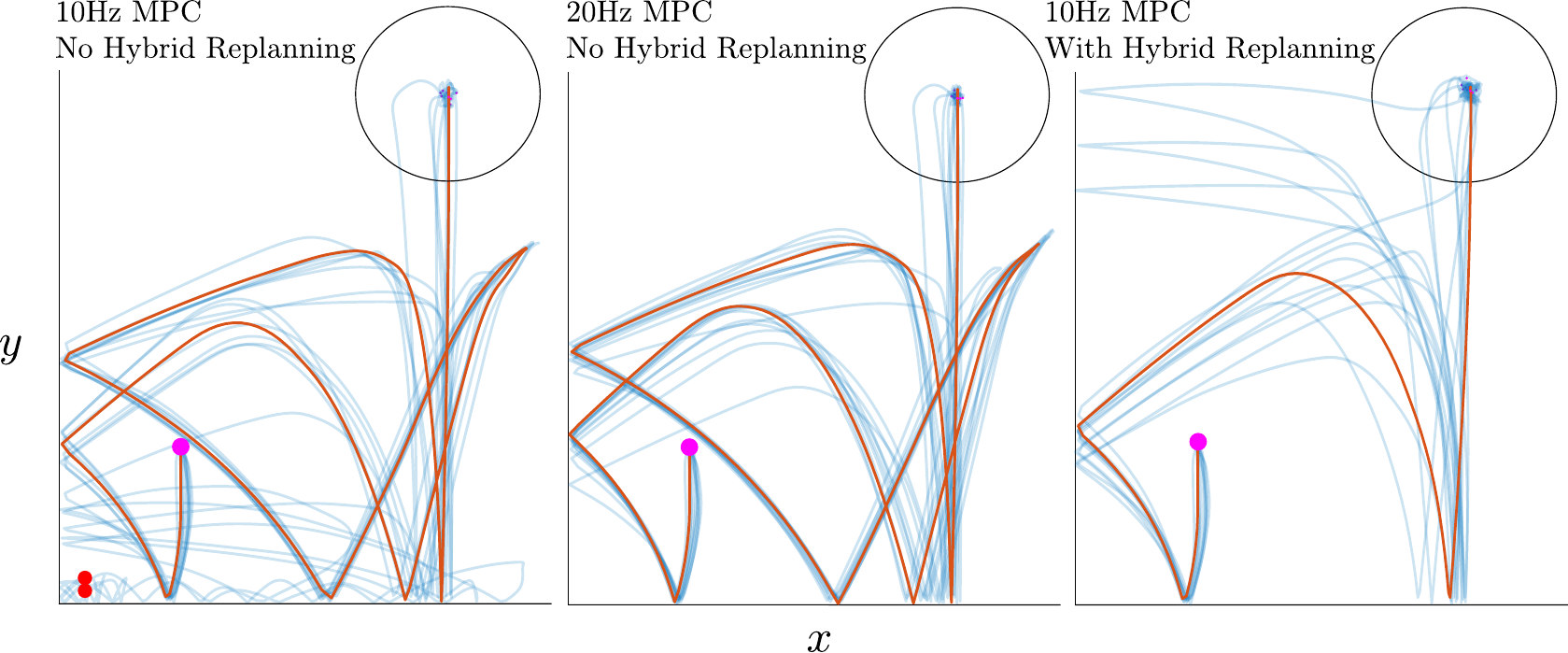}
    \caption{This plot shows the controlled bouncing ball simulated with three different controllers while subject to the same 10 disturbances. The large magenta circles show the initial position, the black lines are the guards, and the red circles indicate unstable trajectories. The orange line shows the trajectory without disturbances. On the left, the MPC is run slow and there is no hybrid re-planning which leads to two unstable trajectories. In the middle, the MPC is run twice as fast and all trajectories stabilize. On the right, the MPC is run slow but hybrid re-planning occurs at each contact; this controller stabilizes all the trajectories. This highlights how the different tools proposed in the paper can lead to a more robust controller.}
    \vspace{-4mm}
    \label{fig:bouncing_ball_disturbances}
\end{figure}

\textbf{Robotic walking.}
The second example considers a five link, full actuated, biped walking up a set of four stairs. The continuous dynamics of the robot are $M(q)\ddot{q} + C(q, \dot{q})\dot{q} + G(q) = \tau$ where $q$ denotes the joint configuration, $M$ is the inertia matrix, $C$ is the Coriolis forces, $G$ is the force of gravity, and $\tau$ is the torque at each joint. We denote the state of the robot as $s = [q, \dot{q}, x, y]$ where $x$ and $y$ give the position of the torso joint of the robot.

Let $r_x(s)$ and $r_y(s)$ denote the $x$ and $y$ positions of the swing foot which are determined from forward kinematics. Each guard is given by a set $\calG_i = \set{s}{r_y(s) = c_i, r_x(s) \in \phi_i}$ where $\phi_i$ gives the set of $x$ coordinates for the step, $c_i$ gives the height of each step, and $i \in \{1,2,3,4\}$. When the guard condition is met a reset map on the joint velocities $\calR(s) := (I - M^{-1}J^T(JM^{-1}J^T)^{-1}J)\dot{q}$ is applied, where $I$ is the identity matrix and $J$ is the jacobian of the swing foot \cite{grizzle_models_2014}. Additionally, the reset map also re-maps the joints between the swing and stance legs.
Fig. \ref{fig:biped_walking} demonstrates the control architecture in use. In this simulation the mass of all the links are 10\% higher than the planning model and additional, random, forces are applied to the torso joint. A continuous time PD controller tracks the MPC trajectory. The hybrid problem is solved by computing the MPC for each potential step option, after removing steps far away, and selecting the best solution. The robot successfully chooses which steps to use, makes contact, and reaches its desired location.

\section{Conclusion}
This work presented a layered MPC approach for robust stabilization of nonlinear hybrid systems. A hybrid MPC determines a feasible domain and guard sequence. A fixed mode MPC uses that to compute a robustly stabilizing solution. A low level controller mitigates disturbances. Each layer of the control architecture helps with the robustness of the system. Many topics warrant future research including the interplay between the two MPC layers and conditions under which the fixed mode MPC goes infeasible.

\begin{figure}
    \vspace{2mm}
    \centering
    \includegraphics[width=0.9\linewidth]{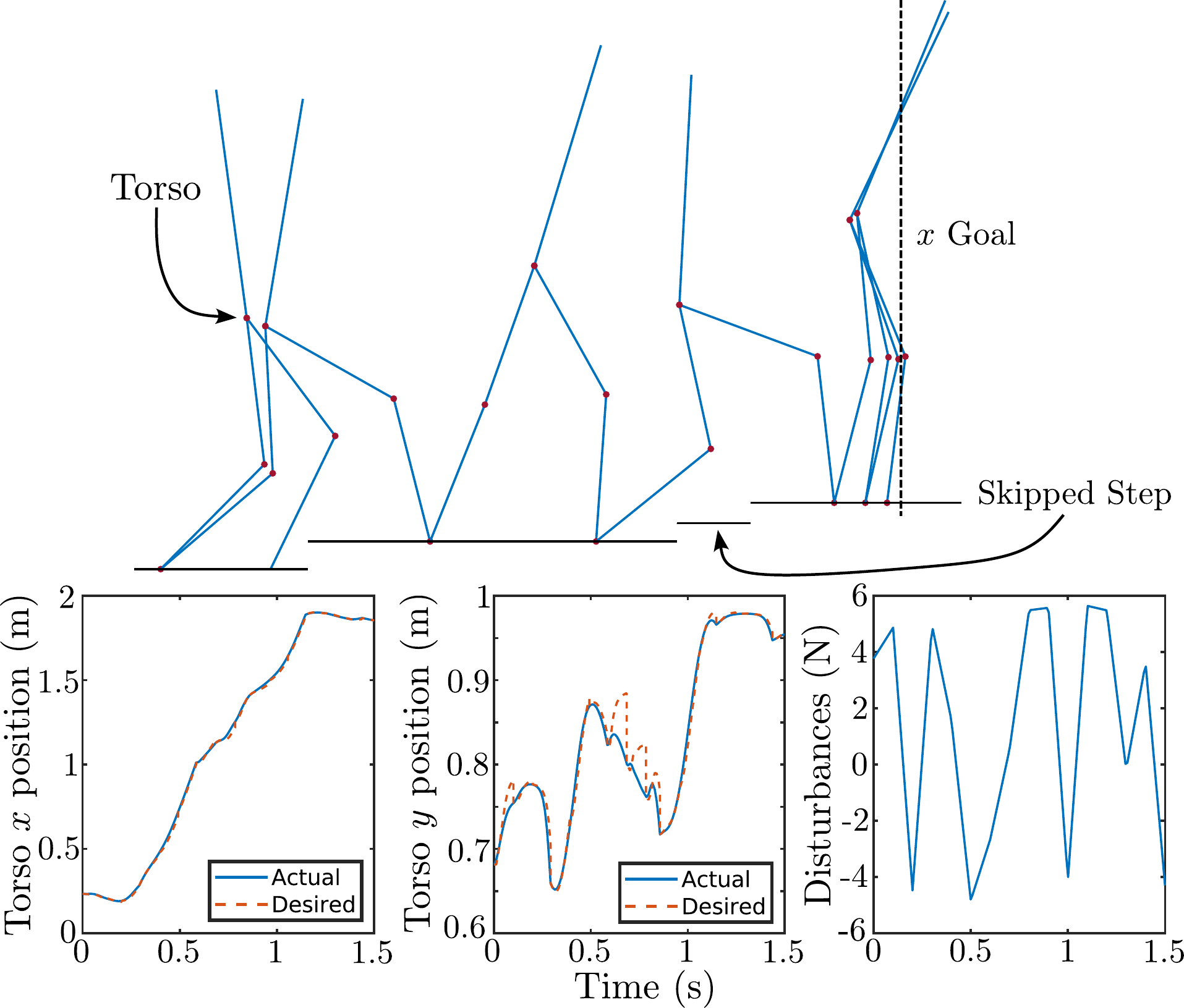}
    \caption{The five link biped traverses a set of stairs with varying lengths and heights. Fixed mode MPC is run at 10Hz and the hybrid MPC is recomputed at each contact. A low level PD controller tracks the MPC target. The biped chooses multiple steps on the longer stair and skips over the shorter stair. The torso joint $x$ and $y$ positions are plotted as well as the disturbance forces applied to the torso in the $x$ direction.
    }
    \vspace{-5mm}
    \label{fig:biped_walking}
\end{figure}

\bibliographystyle{IEEEtran}
\bibliography{Layered_Hybrid_MPC_no_url}
\end{document}